\documentclass[11pt]{article}\textwidth 6.5in\textheight 51pc
\usepackage{amssymb}\usepackage[colorlinks]{hyperref}\usepackage{color}
\usepackage{mathrsfs}\usepackage[normalem]{ulem}\usepackage{amsthm}
\usepackage{amsmath}
\usepackage{amssymb}
\usepackage{verbatim}
\usepackage{braket}
\usepackage{setspace}
\usepackage{natbib}
\usepackage{hyperref}

\DeclareMathOperator{\Tr}{Tr}
\newcommand{\bmu}{\boldsymbol{\mu}}
 
\begin{document}

\newcommand\hreff[1]{\href {http://#1} {\small http://#1}}
\newcommand\trm[1]{{\bf\em #1}}\newcommand\prf{\paragraph{Proof.}}

\newtheorem{thr}{Theorem} 
\newtheorem{lem}{Lemma}
\newtheorem{prp}{Proposition}

\newtheorem{thm}{Theorem} 
\newtheorem{lmm}{Lemma}
\newtheorem{cor}{Corollary}
\newtheorem{con}{Conjecture} 

\newtheorem{blk}{Block}
\newtheorem{dff}{Definition}
\newtheorem{asm}{Assumption}
\newtheorem{rmk}{Remark}
\newtheorem{clm}{Claim}
\newtheorem{example}{Example}

\newcommand\floor[1]{{\lfloor#1\rfloor}}\newcommand\ceil[1]{{\lceil#1\rceil}}

\newcommand{\lea}{<^{+}}
\newcommand{\gea}{>^{+}}
\newcommand{\eqa}{=^{+}}
\renewcommand{\lem}{<^{\ast}}
\newcommand{\gem}{>^{\ast}}
\newcommand{\eqm}{=^{\ast}}
\newcommand{\lel}{<^{\log}}
\newcommand{\gel}{>^{\log}}
\newcommand{\eql}{=^{\log}}

\newcommand{\BB}{\mathbf{BB}}
\newcommand\D{{\mathbf{D}}}
\newcommand\R{\mathbb{R}}
\newcommand\Z{\mathbb{Z}}
\newcommand\C{\mathbb{C}}
\newcommand\M{\mathbf{M}}
\newcommand\h{\mathbb{h}}
\newcommand\N{\mathbb{N}}
\newcommand\BT{\Sigma}
\newcommand\FS{\BT^*}
\newcommand\IS{\BT^\infty}
\newcommand\FIS{\BT^{*\infty}}

\renewcommand\H{{\mathbf H}}

\newcommand\ml{\underline{\mathbf m}}

\newcommand\uhr{\upharpoonright}

\renewcommand\R{{\mathbf R}}
\newcommand\Hq{{\mathbf H}_\mathbf{q}}
\newcommand\Hd{{\mathbf H}_\mathbf{d}}
\newcommand\K{{\mathbf K}} 
\newcommand\QC{{\mathbf {QC}}}
\newcommand\I{{\mathbf I}}
\newcommand\SI{{\mathbf i}}
\newcommand\m{{\mathbf m}}
\renewcommand\d{{\mathbf d}}
\renewcommand\h{{\mathbf h}}
\newcommand\ch{{\chi}}

\author {Samuel Epstein\footnote{\href{mailto:samepst@icloud.com}{samepst@icloud.com}
}}
\title{\vspace*{-3pc} Algorithmic No-Cloning Theorem} \date{\today}\maketitle
\begin{abstract}
	We introduce notions of algorithmic mutual information and deficiency of randomness of quantum states. These definitions enjoy conservation inequalities over unitary transformations and partial traces. We show that a large majority of pure states have minute self algorithmic information. We provide an algorithmic variant to the no-cloning theorem, by showing that only a  small minority of quantum pure states can clone a non negligible amount of algorithmic information. We also provide a chain rule inequality for quantum algorithmic entropy. We show that deficiency of randomness does not increase under POVM measurements.
\end{abstract}
\section{Introduction}
The no-cloning theorem states that every unitary transform cannot clone an arbitrary quantum state. Hiowever some unitary transforms can clone a subset of pure quantum states. For example, given basis states $\ket{1}, \ket{2}, \ket{3}, \dots$ there is a unitary transform that transforms each $\ket{i}\ket{0}$ to $\ket{i}\ket{i}$. In addition, there exists several generalizations to the no-cloning theorem, showing that imperfect clones can be made. In \citep{BuzekHi96}, a universal cloning machine was introduced that can clone an arbitrary state with the fidelity of 5/6.

In this paper, we look at the no-cloning theorem from an algorithmic perspective. We introduce the notion of the algorithmic mutual information, $\I$, between two quantum states. This is a symmetric measure that enjoys conservation inequalities over unitary transforms and partial traces. Quantum algorithmic information upper bounds the amount of classical algorithmic information between POVM measurements of two quantum states.

Given this information function, a natural question to pose is whether a considerable portion of pure states can use a unitary transform to produce two states that share a large amount of information. This paper answers this question in the negative. Only a very sparse set of pure states can, given any unitary transform, duplicate algorithmic information.

This result is achieved in a two step process. In the first step, we show that only a small minority of pure states have non negligible self information. This fact is interesting in it own right, since we show that most pure states have high quantum algorithmic entropy. In the second step, we show that the information between any two states produced from a unitary tranform and the quantum state $\ket{\phi}\ket{0}$ is upper bounded by the self information of $\ket{\phi}$. More specifically,
\begin{enumerate}
	\item Let $\Lambda$ be the uniform distribution on the unit sphere of an $n$ qubit space.\\ $\int 2^{\I(\ket{\phi}:\ket{\phi})}d\Lambda=O(1)$,
	\item If $C\ket{\phi}\ket{0}=\ket{\psi}\ket{\varphi}$ for unitary transform $C$, then $\I({\ket{\psi}}\,{:}\,{\ket{\varphi}})<\I({\ket{\phi}}\,{:}\,{\ket{\phi}})+O(1)$.
\end{enumerate}

The details of the above statements can be found in Theorems \ref{thm:selfinfo} and \ref{thm:nocloning}. These two results combined together imply that on average, states can only duplicate a negligible amount of information. However the basis states, $\ket{i}$, can use a unitary transform to clone at least $\K(i)$ information, where $\K$ is the Kolmogorov complexity measure.

In addition to this algorithmic take on the no-cloning theorem, we provide some other results as well. We define the notion of randomness of one quantum state with respect to another, possibly non computable, quantum state. We show conservation of randomness with respect to unitary transformations and partial traces. We prove a chain rule inequality with respect to quantum algorithmic entropy. We show that POVM measurements do not increase the deficiency of randomness of a quantum state with respect to another quantum state.

\section{Related Work}

The study of Kolmogorov complexity originated from the work of~\citep{Kolmogorov65}. The canonical self-delimiting form of Kolmogorov complexity was introduced in~\citep{ZvonkinLe70} and~\citep{Chaitin75}. The universal probability $\m$ was introduced in~\citep{Solomonoff64}.

More information about the history of the concepts used in this paper can be found in the textbook~\citep{LiVi08}. Quantum algorithmic probability was studied in~\citep{Gacs01}. A type of quantum complexity dependent on descriptive complexity was introduced in\citep{Vitanyi00}. Another variant, quantum Kolmogorov complexity, was developed in\citep{BerthiaumeVaLa01}. Quantum Kolmogorov complexity uses a universal quantum Turing machine. The extension of G\'{a}cs entropy to infinite Hilbert spaces can be seen in \citep{Benatti14}. In \citep{Benatti06}, a quantum version of Brudno's theorem is proven, connecting the Von Neumann entropy rate and two notions of quantum Kolmogorov complexity. In \citep{NiesSc18}, quantum Martin L\"{o}f sequences were introduced.

\section {Conventions and Kolmogorov Complexity Tools.}

Let $\mathbb{N}$, $\Sigma$, $\FS$ be the set of natural numbers, bits and finite sequences. The $i$th bit of a sequence $x\in\FS$ is $x[i]$. $\|x\|=n$ for $x{\in}\Sigma^n$.  $[A]=1$ if statement $A$ holds, else $[A]=0$. ${\lea}f$, ${\gea}f$, ${\eqa}f$, and ${\lel}f$, ${\gel} f$, ${\eql}f$, and ${\lem}f$, ${\gem}f$, and ${\eqm}f$
denote ${<}f{+}O(1)$, ${>}f{-}O(1)$, ${=}f{\pm}O(1)$, and ${<}f{+}O(\log(f
{+}1))$, ${>}f{-}O(\log(f{+}1))$, ${=}f{\pm}O(\log(f{+}1))$, and 
${<}f/O(1)$, ${>}f/O(1)$, ${=}f\,{*}{/}\,O(1)$ respectively. To explicitly specify a constant dependent on parameters $\alpha_1,\alpha_2\dots$, we use the notation $O_{\alpha_1,\alpha_2\dots}(1)$. 

For Turing machine $U$, we say program $p\in\FS$ outputs string $x$, with $U(p)=x$, if $U$ outputs $x$ after reading $\|p\|$ bits of $p$ from the input tape and halts. Otherwise if $U$ reads ${\neq}\|p\|$ bits or it never halts, then $U(p)=\perp$. By this definition $U$ is  a prefix algorithm. Auxiliary inputs $y$ to $U$ are denoted by  $U_y$. Our $U$ is universal,
i.e.\ minimizes (up to $\eqa$) Kolmogorov complexity $\K$. This measure is $\K(x/y)=\min_p\{\|p\|:U_y(p){=}x\}$. The universal probability of an element $x$ relative to string $y$ is $\m(x/y)=
\sum_p2^{{-}\|p\|}[U_y(p) = x]$. We omit empty $y$. By the coding theorem, $-\log\m(x)\eqa\K(x)$. When we say that universal Turing machine is relativized to an elementary object, this means that an encoding of the object is provided to the universal Turing machine on an auxilliary tape.
\newpage
\section{Quantum States}

We deal with finite $n$ dimensional Hilbert spaces $\mathcal{G}_n$, with bases $\ket{\alpha_1},\ket{\alpha_2},\dots,\ket{\alpha_n}$. We assume $\mathcal{G}_{n+1}\supseteq\mathcal{G}_n$ and the bases for $\mathcal{G}_n$ are the beginning of that of $\mathcal{G}_{n+1}$. An $n$ qubit space is denoted by $\mathcal{Q}_n = \bigotimes_{i=1}^n\mathcal{Q}_1$, where qubit space $\mathcal{Q}_1$ has bases $\ket{0}$ and $\ket{1}$. For $x\in\Sigma^n$ we use $\ket{x}\in\mathcal{Q}_n$ to denote $\bigotimes_{i=1}^n\ket{x[i]}$. The space $\mathcal{Q}_n$ has $2^n$ dimensions and we identify it with $\mathcal{G}_{2^n}$.

A pure quantum state $\ket{\phi}$ of length $n$ is represented as a unit vector in $\mathcal{Q}_n$. Its corresponding element in the dual space is denoted by $\bra{\phi}$. The tensor product of two vectors is denoted by $\ket{\phi}\otimes\ket{\psi} = \ket{\phi}\ket{\psi} = \ket{\phi\psi}$. The inner product of $\ket{\psi}$ and $\bra{\phi}$ is denoted by $\braket{\psi|\phi}$.

The transpose of a matrix $M$ is denoted by $M^T$. The tensor product of two matrices is denoted by $A\otimes B$. The trace of a matrix is denoted by $\Tr$ and for tensor product space $\mathcal{G}_X\otimes\mathcal{G}_Y$, the partial trace is denoted by $\Tr_Y$. For positive semidefinite matrices, $\sigma \leq\rho$ iff $\rho-\sigma$ is positive semidefinite. Mixed states are represented by density matrices, which are, self adjoint, positive semidefinite, operators of trace 1. A semi-density matrix has non-negative trace less than or equal to 1. 

A pure quantum state $\ket{\phi}$ and (semi)density matrix $\sigma$ are called \textit{elementary} if their real and imaginary components have rational coefficients. Elementary objects can be encoded into strings or integers and be the output of halting programs. Therefore one can use the terminology $\K(\ket{\phi})$ and $\K(\sigma)$, and also $\m(\ket{\phi})$ and $\m(\sigma)$. Algorithmic quantum entropy, also known as G\'{a}cs entropy, is defined using the following universal semi-density matrix, parametered by $x\in\FS$, with
$$\bmu_{/x} = \sum_{\textrm{elementary }\ket{\phi}\in\mathcal{Q}_n}\m(\ket{\phi}/x,n)\,{\ket{\phi}}{\bra{\phi}}.$$ 
The parameter $n$ represents the number of qubits. We use $\bmu_X$ to denote the matrix $\bmu$ over the Hilbert space denoted by symbol $X$. The G\'{a}cs entropy of a mixed state $\sigma$, conditioned on $x\in\FS$ is defined by $\H(\sigma /x) = \ceil{-\log \Tr\bmu_{/x}\sigma}$. We use the following notation for pure states, with $\H(\ket{\phi}/x) = \H(\,{\ket{\phi}}{\bra{\phi}}\,/x)$. For empty $x$ we omit.  This definition of algorithmic entropy generalizes $\underline{H}$ in \citep{Gacs01} to mixed states.

We say program $q\in\FS$ lower computes positive semidefinite matrix $\sigma$ if, given as input to universal Turing machine $U$, the machine $U$ reads  $\leq\|q\|$ bits and outputs, with or without halting, a sequence of elementary semi-density matrices $\{\sigma_i\}$ such that $\sigma_{i}\leq \sigma_{i+1}$ and $\lim_{i\rightarrow\infty}\sigma_i = \sigma$. A matrix is lower computable if there is a program that lower computes it. The matrix $\bmu$ is universal in that it multiplicatively dominates all lower computable semi-density matrices, as shown in the following theorem, which will be used throughout this paper. \\ \\
\textbf{Theorem} (\citep{Gacs01}, Theorem 2)$ $\\
\label{thm:lowercompute}
\textit{If $q\in\FS$ lower computes semi-density matrix $\sigma$, then $\m(q/n)\sigma \lem \bmu$.}
\newpage

\section{Addition Theorem}
\label{sec:chain}
The addition theorem for classical entropy asserts that the joint entropy
for a pair of random variables is equal to the entropy of one plus the conditional entropy of the other, with $\mathcal{H}(\mathcal{X})+\mathcal{H}(\mathcal{Y}/\mathcal{X})=\mathcal{H}(\mathcal{X},\mathcal{Y})$. For algorithmic entropy, the chain rule is slightly more nuanced, with $\K(x)+\K(y/x,\K(x))\eqa\K(x,y)$. An analogous relationship cannot be true for G\'{a}cs entropy, $\H$, since as shown in Theorem 15 of \citep{Gacs01}, there exists elementary $\ket{\phi}$ where $\H(\ket{\phi}\ket{\phi})-\H(\ket{\phi})$ can be arbitrarily large, and $\H(\ket{\phi}/\ket{\phi})\eqa 0$. However, the following theorem shows that a chain rule inequality does hold for $\H$.

For $n^2\times n^2$ matrix $A$, let $A[i,j]$ be the $n\times n$ submatrix of $A$ starting at position $(n(i-1)+1,n(j-1)+1)$. For example for $n=2$ the matrix $$
A=\begin{bmatrix} 
1 & 2 & 3 & 4\\
5 & 6 & 7 & 8\\
9 & 10 & 11 & 12\\
13 & 14 & 15 &16
\end{bmatrix}
$$
has $A[1,1]=\begin{bmatrix}
1 & 2\\
5 & 6
\end{bmatrix}$, $A[1,2]=\begin{bmatrix}
3 & 4\\
7 & 8
\end{bmatrix}$,  $A[2,1]=\begin{bmatrix}
9 & 10\\
13 & 14
\end{bmatrix}$,  $A[2,2]=\begin{bmatrix}
11 & 12\\
15 & 16
\end{bmatrix}$. 

For $n^2\times n^2$ matrix $A$ and $n\times n$ matrix $B$, let $M_{AB}$ be the $n\times n$ matrix whose $(i,j)$ entry is equal to $\Tr A[i,j]B$. For any $n\times n$ matrix $C$, in can be seen that $\Tr A(C\otimes B)=\Tr M_{AB}C$. Furthermore if $A$ is lower computable and $B$ is elementary, then $M_{AB}$ is lower computable.

For elementary semi density matrices $\rho$, we use $\langle \rho,\H(\rho)\rangle$ to denote the encoding of the pair of an encoded $\rho$ and an encoded natural number $\H(\rho)$.

\begin{thm}[\textbf{Addition Inequality}]
\label{thm:addition}
For semi-density matrices $\sigma$, $\rho$, elementary $\rho$,\\  $\H(\rho)+\H(\sigma / \langle\rho,\H(\rho)\rangle) \lea \H(\sigma\otimes\rho)$.
\end{thm}
\begin{proof}
Let $\bmu_{2n}$ be the universal lower computable semi density matrix over the space of 2n qubits, $\mathcal{Q}_{2n} = \mathcal{Q}_n{\otimes}\mathcal{Q}_n= \mathcal{Q}_A\otimes\mathcal{Q}_B$. Let $\bmu_n$ be the universal matrix of the space over $n$ qubits. We define the following bilinear function over complex matrixes of size $n\times n$, with $T(\nu,\delta) =\Tr\bmu_{2n}(\nu\otimes\delta)$. For fixed $\rho$, $T(\nu,\rho)$ is of the form $T(\nu,\rho) = \Tr M_{\bmu_{2n}\rho}\nu$. The matrix $M_{\bmu_{2n}\rho}$ has trace equal to 
\begin{align*}
\Tr M_{\bmu_{2n}\rho} &= T(\rho,\mathit{I})\\
 &= \Tr\bmu_{2n}(\rho\otimes I)\\
 &= \Tr\left( (\Tr_{\mathcal{Q}_B} \bmu_{2n})\rho\right)\\
 &\eqm \Tr \bmu_n\rho \\
 &\eqm {2^{-\H(\rho)}},
\end{align*}
using Theorem 14 of \citep{Gacs01}, which states $\Tr_Y\bmu_{XY}\eqm \bmu_X$. By the definition of $M$, since $\bmu_{2n}$ and $\rho$ are positive semi-definite, it must be that $M_{\bmu_{2n}\rho}$ is positive semi-definite. Since the trace of $M_{\bmu_{2n}\rho}$ is $\eqm {2^{-\H(\rho)}}$
, it must be that up to a multiplicative constant, $2^{\H(\rho)}M_{\bmu_{2n}\rho}$ is a semi-density matrix. Since $\bmu$ is lower computable and $\rho$ is elementary, by the definition of $M$, $2^{\H(\rho)}M_{\bmu_{2n}\rho}$ is lower computable relative to the string $\langle\rho,\H(\rho)\rangle$. Therefore we have that $2^{\H(\rho)}M_{\bmu_{2n}\rho}\lem \bmu_{/\langle\rho,\H(\rho)\rangle}$. So we have that $-\log \Tr 2^{\H(\rho)}M_{\bmu_{2n}\rho}\sigma = -\H(\rho)-\log T(\sigma,\rho) \eqa \H(\sigma\otimes\rho)-\H(\rho) \gea -\log \bmu_{/(\rho,\H(\rho))}\sigma \eqa \H(\sigma/\langle\rho,\H(\rho)\rangle)$.
\end{proof}

\section{Deficiency of Randomness and Information}
\label{sec:rareinfo}
In this section, we extend algorithmic conservation of randomness and information to the quantum domain. We also present lower and upper bounds for the amount of self algorithmic information that a mixed quantum state can have. 

The classical deficiency of randomness of a semimeasure $\gamma$ with respect to a computable probability measure $P$ is denoted by $\d(\gamma|P) = \log\sum_x \gamma(x)\m(x)/P(x)$. This term enjoys conservation inequalities, where for any computable transform $T:\FS\rightarrow\FS$, $\d(T\gamma|TP) < \d(\gamma|P)+O_{T,P}(1)$. 

For semi-density matrix $\rho$, a matrix $\nu$ is a $\rho$-\textit{test}, $\nu\in\mathcal{T}_\rho$, if it is lower computable and $\Tr\nu\rho\leq 1$. In \citep{Gacs01}, the universal randomness test of with respect to elementary $\rho$ was defined as $\nu_\rho=\sum_i\m(i)\nu_i$, where $\{\nu_i\}$ is an enumeration of $\mathcal{T}_\rho$. Paralleling the classical definition, the deficiency of randomness of $\sigma$ with respect to $\rho$ was defined as $\log\Tr\nu_\rho\sigma$. 

For non computable $\rho$, $\mathcal{T}_\rho$ is not necessarily enumerable, and thus a universal lower computable randomness test does not necessarily exist, and cannot be used to define the $\rho$ deficiency of randomness. So in this case, the deficiency of randomness is instead defined using an aggregation of $\rho$-tests, weighted by their lower algorithmic probabilities. This is reminiscient of the definition of $\mathbf{D}$ in \citep{Levin84}, which is an aggregation of integral tests, weighted by their algorithmic probabilities. The lower algorithmic probability of a lower computable matrix $\sigma$ is $\ml(\sigma/x) = \sum \{\m(q/x)\,{:}\,q\textrm{ lower computes }\sigma\}$. Let $\mathfrak{T}_\rho = \sum_{\nu\in\mathcal{T}_\rho} \ml(\nu/n)\nu$.

\begin{dff} 
 The deficiency of randomness of $\sigma$ with respect to $\rho$ is $\d(\sigma|\rho) = \log \Tr\mathfrak{T}_\rho\sigma$.
\end{dff}
By definition, $\mathfrak{T}_\rho$ is universal, since for every lower computable $\rho$-test $\nu$, $\ml(\nu)\nu<\mathfrak{T}_\rho$. So, relativized to invertible elementary $\rho$, by Theorem 17 of \citep{Gacs01}, $\mathfrak{T}_\rho$ is equal, up to a multiplicative constant to the universal lower computable $\rho$ test, and also $\d(\sigma|\rho)\eqa \log\Tr \rho^{-1/2}\bmu\rho^{-1/2 }\sigma$. This parallels the classical definition of $\d(x|P)=\log\m(x)/P(x)$.

\begin{prp}
\label{prp:lowerprob}
For semi-density matrix $\nu$, relativized to unitary transform $A$, $\ml(A^*\nu A/n)\gem \ml(\nu/n)$.
\end{prp}
\begin{proof}
For every string $q$ that lower computes $\nu$, there is a string $q_A$ of the form $rq$, that lower computes $A^*\nu A$. 
This string $q_A$ uses the helper code $r$, and $\langle A\rangle$ on the auxilliary tape, to take the intermediary outputs $\xi_i$ of $q$ and outputs the intermediary output $A^*\xi_i A$. Since the complexity of $r$ is a constant, $\m(q_A/n)\gem \m(q/n)$.
\begin{align*}
\ml(\nu/n) &=\sum\{\m(q/n):q\textrm{ lower computes } \nu\}\\
&\lem \sum\{\m(q_A/n):q\textrm{ lower computes } \nu\}\\
& \lem\sum\{\m(q'/n):q'\textrm{ lower computes } A^*\nu A\}\\
&\lem \ml(A^*\nu A/n).
\end{align*}
\end{proof}
\newpage
\begin{thm}[Conservation of Randomness, Unitary Transform]
\label{thm:ConsRan}
For semi-density matrices $\sigma$ and $\rho$, relativized to  elementary unitary transform $A$, $\d(A\sigma A^*|A\rho A^*) \eqa \d(\sigma|\rho).$
\end{thm}
\begin{proof}
If $\nu\in\mathcal{T}_{A\rho A^*}$, then $A^*\nu A\in\mathcal{T}_\rho$. This is because by assumption $\Tr \nu A\rho A^*\leq 1$. So by the cyclic property of trace $\Tr A^*\nu A\rho\leq 1$. Therefore since $A^*\nu A$ is lower computable, $A^*\nu A\in\mathcal{T}_{\rho}$. From proposition \ref{prp:lowerprob}, $\ml(A^*\nu A/n)\gem \ml(\nu/n)$. So we have the following inequality
\begin{align*} 
\d(A\sigma A^*|A\rho A^*) &= \log \sum_{\nu\in\mathcal{T}_{A\rho A^*}}\ml(\nu/n)\Tr\nu A\sigma A^*\\
&\lea \log\sum_{\nu\in\mathcal{T}_{A\rho A^*}}\ml(A^*\nu A/n)\Tr A^*\nu A\sigma\\
&\lea \d(\sigma|\rho). 
\end{align*}
The other inequality follows from using the above reasoning with $A^*$, $A\sigma A^*$, and $A\rho A^*$.
\end{proof}
Conservation of randomness occurs also over a partial trace, as shown in the following theorem. Deficiency of randomness decreases with respect to the reduced quantum states.
\begin{thm}[Conservation of Randomness, Partial Trace]
For $m{<}n$, for the space of $n$ qubits, $\mathcal{Q}_{n}=\mathcal{Q}_{n-m}\otimes\mathcal{Q}_{m}$, relativized to $m$ and $n$, $\d(\Tr_{\mathcal{Q}_{m}}\sigma|\Tr_{\mathcal{Q}_{m}}\rho)\lea \d(\sigma|\rho)$.
\end{thm}
\begin{proof}
If $\nu\in\mathcal{T}_{\Tr_{\mathcal{Q}_{m}}\rho}$, then $\nu\otimes I_{m}\in\mathcal{T}_\rho$, where $I_{m}$ is the identity operator over $\mathcal{Q}_{m}$. This is because $1\geq\Tr \nu\Tr_{\mathcal{Q}_{m}}\rho=\Tr (\nu\otimes I_{m})\rho$. Since $\nu\otimes I_{m}$ is lower computable, $\nu\otimes I_{m}\in\mathcal{T}_\rho$. Also $\ml(\nu\otimes I_{m})\gem \ml(\nu)$. So
\begin{align*}
\d(\Tr_{\mathcal{Q}_{m}}\sigma|\Tr_{\mathcal{Q}_{m}}\rho) &=\log\sum_{\nu\in \mathcal{T}_{\Tr_{\mathcal{Q}_{m}}\rho}}\ml(\nu)\Tr\nu\Tr_{\mathcal{Q}_{m}}\sigma,\\
&\lea \log\sum_{\nu\in \mathcal{T}_{\Tr_{\mathcal{Q}_{m}}\rho}}\ml(\nu\otimes I_{m})\Tr(\nu\otimes I_{m})\sigma,\\
&\lea \d(\sigma|\rho).
\end{align*}
\end{proof}

\subsection{Information}
For a pair of random variables, $\mathcal{X}$, $\mathcal{Y}$, their mutual information is defined to be $\I(\mathcal{X}:\mathcal{Y})=\mathcal{H}(\mathcal{X})+\mathcal{H}(\mathcal{Y})-\mathcal{H}(\mathcal{X},\mathcal{Y})=\mathcal{H}(\mathcal{X})-\mathcal{H}(\mathcal{X}/\mathcal{Y})=\sum_{x,y}p(x,y)\log p(x,y)/p(x)p(y)$. This represents the amount of correlation between $\mathcal{X}$ and $\mathcal{Y}$. Another intrepretation is that the mutual information between $\mathcal{X}$ and $\mathcal{Y}$ is the reduction in uncertainty of $\mathcal{X}$ after being given access to $\mathcal{Y}$.

Quantum mutual information between two subsystems described by states $\rho_A$ and $\rho_B$ of a composite system described by a joint state $\rho_{AB}$ is $I(A:B)=S(\rho_A)+S(\rho_B)-S(\rho_{AB})$, where $S$ is the Von Neumman entropy. Quantum mutual information measures the correlation between two quantum states.

The algorithmic information between two strings is defined to be $\I(x:y/c)=\K(x/c)+\K(y/c)-\K(x,y/c)$. By definition, it measures the amount of compression two strings achieve when grouped together.

The three definitions above are based off the difference between a joint aggregate and the separate parts. Another approach is to define information between two semi-density matrices as the deficiency of randomness over $\bmu\otimes\bmu$, with the mutual information of $\sigma$ and $\rho$ being $\d(\sigma\otimes\rho|\bmu\otimes\bmu)$.  This is a counter argument for the hypothesis that the states are independently chosen according to the universal semi-density matrix $\bmu$. This parallels the classical algorithmic case, where $\I(x:y)\eqa \d((x,y)|\m\otimes\m)\eqa \K(x)+\K(y)-\K(x,y)$. In fact, using this definition, all the theorems in Section \ref{sec:rareinfo} can be proven. However to achieve the conservation inequalities in Section \ref{sec:nocloning}, a further refinement is needed, with the restriction of the form of the $\bmu\otimes\bmu$ tests. Let $\mathcal{C}_{C\otimes D}$ be the set of all lower computable matrices $A\otimes B$, such that $\Tr(A\otimes B)(C\otimes D)\leq 1$. Let $\mathfrak{C}_{C\otimes D}=\sum_{A\otimes B\in\mathcal{C}_{C\otimes D}}\ml(A\otimes B/n)A\otimes B$.

\begin{dff} The mutual information between two semi-density matrices $\sigma$, $\rho$ is defined to be $\I(\sigma\,{:}\,\rho)=\log\Tr\mathfrak{C}_{\bmu\otimes\bmu}(\sigma\otimes\rho)$.
\end{dff}
Up to an additive constant, information is symmetric. 
\begin{thm}
$\I(\sigma\,{:}\,\rho)\eqa\I(\rho\,{:}\,\sigma)$.
\end{thm}
\begin{proof}
This follows from the fact that for every $A\otimes B\in\mathcal{C}_{\bmu\otimes\bmu}$, the matrix $B\otimes A\in\mathcal{C}_{\bmu\otimes\bmu}$. Furthermore, since $\ml(A\otimes B/n)\eqm\ml(B\otimes A/n)$, this guarantees that $\Tr\mathfrak{C}_{\bmu\otimes\bmu}(\sigma\otimes\rho)\eqm\Tr\mathfrak{C}_{\bmu\otimes\bmu}(\rho\otimes\sigma)$, thus proving the theorem.
\end{proof}

Classical algorithmic information non-growth laws asserts that the information between two strings cannot be increased by more than a constant depending on the computable transform $f$, with $\I(f(x):y)<\I(x:y)+O_f(1)$. Conservation inequalities have been extended to probabilistic transforms and infinite sequences. The following theorem shows information non-growth in the quantum case; information cannot increase under an elementary unitary transform. The general form of the proof to this theorem is analogous to the proof of Corollary 1 in \citep{Levin84}.

\begin{thm}[Conservation of Information, Unitary transform]
\label{thm:consinfo}
For semi-density matrices $\sigma$ and $\rho$, relativized to elementary unitary transform $A$, $\I(A\sigma A^*\,{:}\,\rho)\eqa \I(\sigma\,{:}\,\rho)$.
\end{thm}
\begin{proof}

Given density matrices $A$, $B$, $C$ and $D$, we define $\d'(A\otimes B|C\otimes D)=\log\mathfrak{C}_{C\otimes D}A\otimes B$. Thus $\I(\sigma:\rho)=\d'(\sigma\otimes \rho |\bmu\otimes\bmu)$. The semi-density matrix $A\bmu A^*$ is lower semicomputable, so therefore $A\bmu A^* \lem \bmu$ and also $(A\bmu A^*\otimes\bmu)\lem \bmu\otimes\bmu$. So  if $E\otimes F\in\mathcal{C}_{\bmu\otimes\bmu}$ then there is a positive constant $c$, where $c(E\otimes F) \in \mathcal{C}_{A\bmu A^*\otimes\bmu}$. So we have
\begin{align*}
\d'(A\sigma A^*\otimes\rho |\bmu\otimes\bmu)
&=\log\sum_{E\otimes F\in\mathcal{C}_{\bmu\otimes\bmu}}\ml(E\otimes F/n)\Tr (E\otimes F) (A\sigma A^*\otimes \rho)\\
& \lea \log\sum_{E\otimes F\in\mathcal{C}_{\bmu\otimes\bmu}}\ml(c(E\otimes F)/n)\Tr c(E\otimes F) (A\sigma A^*\otimes \rho)\\
&  \lea \d'(A\sigma A^*\otimes \rho |A\bmu A^*\otimes\bmu).
\end{align*}
Using the reasoning of Theorem \ref{thm:ConsRan} on the unitary transform $A\otimes \mathit{I}$ and $\d'$ we have that $\d'(A\sigma A^*\otimes \rho | A\bmu A^*\otimes \bmu) \lea \d'(\sigma\otimes\rho | \bmu\otimes\bmu)$. Therefore we have that $\I(A\sigma A^*\,{:}\,\rho)=\d'(A\sigma A^*\otimes\rho |\bmu\otimes\bmu)\lea\d'(A\sigma A^*\otimes \rho | A\bmu A^*\otimes \bmu)\lea \d'(\sigma \otimes \rho |\bmu \otimes\bmu)\eqa\I(\sigma\,{:}\,\rho)$. The other inequality follows from using the same reasoning with $A^*$ and $A\sigma A^*$. \end{proof}

\subsection{Self Information}

For classical algorithmic information, $\I(x\,{:}\,x)\eqa\K(x)$, for all $x\in\FS$. As shown in this section, this property differs from the quantum case, where there exists quantum states with high descriptional complexity and negligible self information. In fact this is the case for most pure states. The following theorem states that the information between two elementary states is not more than the combined length of their descriptions.

\begin{thm} For elementary $\rho$ and $\sigma$,
$\I(\rho:\sigma)\lea \K(\rho/n)+\K(\sigma/n)$.
\label{thm:infouppperlower}
\end{thm}
\begin{proof}

Assume not. Then for any positive constant $c$, there exists semi-density matrices $\rho$ and $\sigma$, such that $c\m(\rho/n)\m(\sigma/n)2^{\I(\rho:\sigma)}= c\Tr\m(\rho/n)\m(\sigma/n)\mathfrak{C}_{\bmu\otimes\bmu}(\rho\otimes\sigma)>1$. By the definition of $\bmu$, $\m(\rho/n)\rho\lem \bmu$ and $\m(\sigma/n)\sigma\lem\bmu$. Therefore by the definition of the Kronecker product, there is some positive constant $d$ such that for all $\rho$ and $\sigma$, $d\m(\rho/n)\m(\sigma/n)(\rho\otimes\sigma)< (\bmu \otimes \bmu)$, and similarly $d\Tr \m(\rho/n)\m(\sigma/n)\mathfrak{C}_{\bmu\otimes\bmu}(\rho\otimes \sigma)<\Tr \mathfrak{C}_{\bmu\otimes\bmu}(\bmu\otimes\bmu)$. By the definition of $\mathfrak{C}$, it must be that
$\Tr\mathfrak{C}_{\bmu\otimes\bmu}\bmu\otimes\bmu\leq 1$. However for $c=d$, there exists a $\rho$ and a $\sigma$, such that $\Tr\mathfrak{C}_{\bmu\otimes\bmu}\bmu\otimes\bmu > d\Tr\m(\rho/n)\m(\sigma/n)\mathfrak{C}_{\bmu\otimes\bmu}(\rho\otimes\sigma) > 1$, causing a contradiction.

\end{proof}

\begin{thm}
\label{thm:selfinfo}
Let $\Lambda$ be the uniform distribution on the unit sphere of $\mathcal{H}_{2^n}$.
\begin{enumerate}
\item $\H(2^{-n}I)\eqa n$,
\item $\I(2^{-n}I\,{:}\,2^{-n}I)\lea 0$,
\item $\displaystyle\int 2^{-\H(\ket{\psi})}d\Lambda \eqm 2^{-n}$,
\item $\displaystyle\int 2^{\I(\ket{\psi}\,{:}\,\ket{\psi})}d\Lambda\lea 0$.
\end{enumerate} 
\end{thm}

\begin{proof}
(1) follows from $\H(2^{-n}I) \eqa -\log\Tr \bmu 2^{-n}I \eqa n-\log\Tr\bmu\eqa n$. (2) is due to Theorem \ref{thm:infouppperlower}, with $\I(2^{-n}I,2^{-n}I)\lea 2\K(2^{-n}I/n) \lea 0$. (3) and (4) use \citep{Gacs01} Section 5 and \citep{BerthiaumeVaLa01} Section 6.3, with $\int \ket{\psi}\bra{\psi}d\Lambda = 2^{-n}I$ and $\int \,{\ket{\psi}}{\bra{\psi}}\otimes {\ket{\psi}\bra{\psi}}\,d\Lambda = \int \,{\ket{\psi\psi}}{\bra{\psi\psi}}\,d\Lambda$ $= {2^n+1\choose 2}^{-1}P$, where $P$ is the projection onto the space of pure states ${\ket{\psi\psi}}$. So $\int 2^{-\H(\ket{\psi})}d\Lambda \eqm \int \Tr\bmu \ket{\psi}\bra{\psi}d\Lambda \eqm \Tr\bmu\int \ket{\psi}\bra{\psi}d\Lambda\eqm 2^{-n}$, and 
\begin{align*}
\int 2^{\I(\ket{\psi}\,{:}\,\ket{\psi})}d\Lambda &= \int \Tr\mathfrak{C}_{\bmu\otimes\bmu}\,{\ket{\psi}}{\bra{\psi}}\otimes {\ket{\psi}\bra{\psi}}\,d\Lambda\\
&=\Tr\mathfrak{C}_{\bmu\otimes\bmu} \int{\ket{\psi}}{\bra{\psi}}\otimes {\ket{\psi}\bra{\psi}}\,d\Lambda\\ 
&= \Tr\mathfrak{C}_{\bmu\otimes\bmu}{2^n+1\choose 2}^{-1}P\\
&\lem \Tr\mathfrak{C}_{\bmu\otimes\bmu}2^{-2n}I\\
&\eqm 2^{\I(2^{-n}I:2^{-n}I)}\\
&\lea 0.
\end{align*}
\end{proof}

\subsection{Measurements}
A POVM $E$ is a finite or infinite set of positive definite matrices $\{E_k\}$ such that $\sum_k E_k = I$. For a given semi-density matrix $\sigma$, a POVM $E$ induces a semi measure over integers, where $E\sigma(k) = \Tr E_k\sigma$. This can be seen as the probability of seeing measurement $k$ given quantum state $\sigma$ and measurement $E$. An elementary POVM $E$ has a program $q$ such that $U(q)$ outputs an enumeration of $\{E_k\}$, where each $E_k$ is elementary. Theorem \ref{thm:povm} shows that measurements can increase only up to a constant factor, the deficiency of randomness of a quantum state. Note that the $\d(E\sigma|E\rho)$ term represents the classical deficiency of randomness of a semimeasure with respect to a computable probability measure, as defined in the beginning of Section \ref{sec:rareinfo}.
 
\begin{thm}
\label{thm:povm}
For semi-density matrices $\sigma$, $\rho$, relativized to elementary $\rho$ and POVM $E$,\\  $\d(E\sigma|E\rho)\lea \d(\sigma |\rho)$.
\end{thm}
\begin{proof}
$2^{\d(E\sigma|E\rho)}=\sum_k (\Tr  E_k \sigma)\m(k)/(\Tr E_k \rho)=\Tr (\sum_k (\m(k)/\Tr E_k\rho)E_k)\sigma=\Tr\nu\sigma$, where the matrix $\nu = (\sum_k (\m(k)/\Tr E_k\rho)E_k)$ has $\nu\in\mathcal{T}_\rho$, since $\nu$ is lower computable and $\Tr\nu\rho\leq 1$. So $2^{\d(\sigma|\rho)}\geq\ml(\nu/n)\Tr\nu\sigma = \ml(\nu/n)2^{\d(E|\sigma|E\rho)}$. Since $\ml(\nu/n)\gem 1$, $\d(E\sigma|E\rho)\lea \d(\sigma |\rho)$.
\end{proof}

The information between two quantum states is lower bounded by the classical information of two measurements of those states, as shown in the following theorem. 
The following theorem also implies that pure states that are close to simple rotations of complex basis, i.e. unentangled, states will have high self information. However such states are sparse. On average, a quantum state will have negligible self information.
\begin{thm}
\label{thm:twomeasurements}
For semi density matrices $\rho$, $\sigma$, and $i,j\in \mathbb{N}$, relativized to elementary POVM $E$, \\ $\I(i:j/n)+\log E\sigma(i)E\rho(j)-\K(\I(i:j/n)/n)\lea \I(\sigma:\rho)$.
\end{thm}
\begin{proof}
Since  $z(k)=\Tr\bmu E_k$ is lower semi-computable and $\sum_kz(k)< 1$, $\m(k/n) \gem \Tr\bmu E_k$, and so $1 > 2^{\K(k/n)-O(1)}\Tr\bmu E_k$. Let $\nu=2^{\K(i/n)+\K(j/n)-O(1)}(E_i\otimes E_j)$. $\nu\in\mathcal{T}_{\bmu\otimes\bmu}$, because it is lower semicomputable and $\Tr(\bmu\otimes\bmu)\nu= \Tr2^{\K(i/n)-O(1)}\bmu E_i\otimes2^{\K(j/n)-O(1)}\bmu E_j<1$. Therefore
\begin{align*}
\I(\sigma:\rho) &=\log\Tr\mathfrak{C}_{\bmu\otimes\bmu}(\sigma\otimes\rho)> \log \Tr\ml(\nu/n)\nu(\sigma\otimes\rho)\\
&\gea \K(i/n)+\K(j/n)+\log\Tr E_i\sigma\otimes E_j\rho-\K((\K(i/n)+\K(j/n)),i,j,E/n)\\
&\gea \K(i/n)+\K(j/n)+\log E\sigma(i)E\rho(j)\\
&\;\;\;\;-\K((\K(i/n)+\K(j/n)-\K(i,j/n)),\K(i,j/n),i,j/n)\\
&\gea \K(i/n)+\K(j/n)+\log E\sigma(i)E\rho(j)-\K(\K(i,j/n),i,j/n)-\K(\I(i:j/n)/n)\\
&\gea \I(i:j/n)+\log E\sigma(i)E\rho(j)-\K(\I(i:j/n)/n).
\end{align*}
\end{proof}
\begin{cor}
For semi density matrices $\rho$, $\sigma$, relativized to elementary POVM $E$, \\ $\log\sum_{i,j}2^{\I(i:j/n)-\K(\I(i:j/n)/n)}E\sigma(i)E\rho(j)\lea \I(\sigma:\rho)$.
\end{cor}
This follows from $\I(\sigma:\rho)\gea\log\sum_{i,j} \m(A_{i,j}/n)A_{i,j}(\sigma\otimes\rho)$, where $A_{i,j}\eqm 2^{\K(i/n)+\K(j/n)}(E_i\otimes E_j)$. The reasoning of Theorem \ref{thm:twomeasurements} can then be used.
\begin{cor}
	\label{cor:basis}
	For basis state $\ket{i}$, $\K(i/n)\lea\I(\ket{i}:\ket{i})$.
\end{cor}
Corollary \ref{cor:basis} shows that for basis states $\ket{i}$, a unitary transform that produces $\ket{i}\ket{i}$ from each $\ket{i}\ket{0}$ will duplicate at least $\K(i)$ quantum algorithmic information.

\section{Algorithmic No-Cloning Theorem}
\label{sec:nocloning}
We show that the amount of quantum algorithmic information that can be replicated is bounded by the amount of self information that a state has. As shown in Theorem \ref{thm:selfinfo}, the amount of pure states with high self information is very small. The following theorem states that information non growth holds with respect to partial traces. 
\begin{thm}[Conservation of Information, Partial Trace]
	\label{thm:sinfoconsptrace}
	For $m<n$, and the space of $n$ qubits $\mathcal{Q}_n=\mathcal{Q}_{n-m}\otimes \mathcal{Q}_{m}$, relativized to $m$ and $n$, $\I(\Tr_{\mathcal{Q}_m}\sigma:\Tr_{\mathcal{Q}_m}\rho)\lea\I(\sigma:\rho)$.
\end{thm}
\begin{proof}
	There is a positive constant $c$ where if $A\otimes B$ is in $\mathcal{C}_{\bmu_{n-m}\otimes\bmu_{n-m}}$ then $c(A\otimes I_m)\otimes(B\otimes I_m)$ is in $\mathcal{C}_{\bmu_{n}\otimes\bmu_{n}}$, where $I_{m}$ is the identity operator over $\mathcal{Q}_{m}$. Using Theorem 14 of \citep{Gacs01} which states $\Tr_{\mathcal{Q}_m}\bmu_{n}\eqm\bmu_{n-m}$, we have that $1\geq\Tr(A\otimes B)(\bmu_{n-m}\otimes \bmu_{n-m})=\Tr(A\bmu_{n-m})\otimes (B\bmu_{n-m})\geq\Tr c((A\otimes I_m)\bmu_{n})\otimes((B\otimes I_m)\bmu_{n})=\Tr c((A\otimes I_m)\otimes (B\otimes I_m))(\bmu_n\otimes\bmu_n)$. It is easy to see that $\ml(A\otimes B)\lem \ml((A\otimes I_m)\otimes (B\otimes I_m))$. So
	\begin{align*}
	&\;\;\;\;\;\I(\Tr_{\mathcal{Q}_m}\sigma:\Tr_{\mathcal{Q}_m}\rho)\\
	&=\log\sum_{A\otimes B\in\mathcal{C}_{\bmu_{n-m}\otimes\bmu_{n-m}}}\ml(A\otimes B)\Tr(A\otimes B)(\Tr_{\mathcal{Q}_m}\sigma\otimes\Tr_{\mathcal{Q}_m}\rho)\\
	&=\log\sum_{A\otimes B\in\mathcal{C}_{\bmu_{n-m}\otimes\bmu_{n-m}}}\ml(A\otimes B)\Tr(A\Tr_{\mathcal{Q}_m}\sigma)\otimes(B\Tr_{\mathcal{Q}_m}\rho)\\
	&=\log\sum_{A\otimes B\in\mathcal{C}_{\bmu_{n-m}\otimes\bmu_{n-m}}}\ml(A\otimes B)\Tr((A\otimes I_m)\sigma)\otimes((B\otimes I_m)\rho)\\
	&\lea\log\sum_{A\otimes B\in\mathcal{C}_{\bmu_{n-m}\otimes\bmu_{n-m}}}\ml(c(A\otimes I_m)\otimes (B\otimes I_m))\Tr c((A\otimes I_m)\otimes(B\otimes I_m))(\sigma\otimes\rho)\\
	&\lea \I(\sigma:\rho).
	\end{align*}
\end{proof}
\begin{cor}
\label{cor:infopartialtrace}
	For a density matrix $\sigma$ over the space of $2n$ qubits $\mathcal{Q}_{2n}= \mathcal{Q}_{n}\otimes\mathcal{Q}_{n}=\mathcal{Q}_{A}\otimes\mathcal{Q}_{B}$,
	$\I(\Tr_A\sigma:\Tr_B\sigma)\lea\I(\sigma:\sigma)$.
\end{cor}
\begin{lmm}
\label{lem:mlMAB}
For lower computable semi-density $n^2\times n^2$ matrix $A$ and  elementary semi-density $n\times n$ matrix $B$, $\ml(A/2n)\lem \ml(M_{AB}/n)/\m(B/n)$.
\end{lmm}
\begin{proof}
For semi density matrices, $C$, $D$ of sizes $n^2\times n^2$ and semi-density matrix $E$ of sizes $n\times n$, if $C\leq D$, then $M_{CE}\leq M_{DE}$. This is because for all positive semi-definite matrix $F$ of size $n\times n$, $\Tr M_{DE}F-\Tr M_{CE}F = \Tr (D(F\otimes E)- C(F\otimes E)) = \Tr(D-C)(F\otimes E) \geq 0$. Therefore, for every string $q$ that lower computes $A$, there is a string $q'$ of the form $r\langle B\rangle q$ that uses the helper code $r$ to take the intermediary outputs $\xi_i$ of $q$ and output an intermediary matrix $M_{\xi_iB}$. The limit of $q'$ is $M_{AB}$. So $\m(q/2n)\m(B/n)\lem \m(q'/n)$.
\begin{align*}
\ml(A/n) &= \sum\{\m(q/n):q\textrm{ lower computes } A\}\\
&\lem \sum \{\m(q'/n)/\m(B/n):q\textrm{ lower computes } A\}\\
&\lem (1/\m(B/n))\sum \{\m(s/n):s\textrm{ lower computes } M_{AB}\}\\
&\lem \ml(M_{AB}/n)/\m(B/n).
\end{align*}
\end{proof}

\begin{thm}
	\label{thm:factor}
	For density matrices $\sigma$ and $\rho$, and elementary density matrices $\nu$, and $\xi$ over $n$ qubits, $\I(\sigma\otimes\nu:\rho\otimes\xi)\lea\I(\sigma:\rho)+2\K(\nu/n)+2\K(\xi/n)$. 
\end{thm}
\begin{proof}
Let $c=\Theta(1)2^{-\K(\nu/n)-\K(\xi/n)}$. If $E\otimes F\in\mathcal{C}_{\bmu_{2n}\otimes\bmu_{2n}}$ then $c(M_{E\nu}\otimes M_{F\xi})\in\mathcal{C}_{\bmu_{n}\otimes\bmu_{n}}$, where $M$ is the matrix defined in Section \ref{sec:chain}. This is because $c(M_{E\nu}\otimes M_{F\xi})$ is lower computable and $\Tr c(M_{E\nu}\otimes M_{F\xi}))(\bmu_n\otimes \bmu_n)=c\Tr(E\otimes F)((\bmu_n\otimes\nu)\otimes(\bmu_n\otimes\xi))\leq \Tr(E\otimes F)(\bmu_{2n}\otimes\bmu_{2n})\leq 1$. Furthermore, we use lemma \ref{lem:mlMAB}, where for lower computable $A$, elementary $B$, $\ml(A)\lem \ml(M_{AB})/\m(B)$. So
\begin{align*}
&\;\;\;\;\;\I(\sigma\otimes\nu:\rho\otimes\xi) \\
&= \log\sum_{E\otimes F\in\mathcal{C}_{\bmu_{2n}\otimes\bmu_{2n}}}\ml(E\otimes F/2n)\Tr(E\otimes F)((\sigma\otimes\nu)\otimes(\rho\otimes\xi))\\
&= \log\sum_{E\otimes F\in\mathcal{C}_{\bmu_{2n}\otimes\bmu_{2n}}}\ml(E\otimes F/2n)\Tr(E(\sigma\otimes\nu)\otimes F(\rho\otimes\xi))\\
&= \log\sum_{E\otimes F\in\mathcal{C}_{\bmu_{2n}\otimes\bmu_{2n}}}\ml(E\otimes F/2n)\Tr(M_{E\nu}\sigma\otimes M_{F\xi}\rho)\\
&\lea \log\sum_{E\otimes F\in\mathcal{C}_{\bmu_{2n}\otimes\bmu_{2n}}}\ml(c(M_{E\nu}\otimes M_{F\xi})/n)\Tr(c(M_{E\nu}\sigma\otimes M_{F\xi}\rho))+2\K(\nu/n)+2\K(\xi/n)\\
&\lea \I(\sigma:\rho)+2\K(\nu/n)+2\K(\xi/n).
\end{align*}
\end{proof}
\begin{thm}
\label{thm:nocloning}
	Let $C\ket{\psi}\ket{0^n}=\ket{\phi}\ket{\varphi}$, where $C$ is an elementary unitary transform. Relativized to $C$, $\I(\ket{\phi}:\ket{\varphi})\lea\I(\ket{\psi}:\ket{\psi})$.
\end{thm}
\begin{proof}
We have the inequalities $\I(\ket{\phi}:\ket{\varphi})\lea\I(\ket{\phi}\ket{\varphi}):\ket{\phi}\ket{\varphi})\lea\I(\ket{\psi}\ket{0^n}:\ket{\psi}\ket{0^n})\lea\I(\ket{\psi}:\ket{\psi})$,
supported by Corollary \ref{cor:infopartialtrace} and Theorems \ref{thm:consinfo} and \ref{thm:factor}.
\end{proof}
\section{Discussion}
There are still many open problems with respect to algorithmic quantum deficiency of randomness and information. One question is whether a quantum state has maximized information with itself. More specifically, given density matrices $\rho$ and $\sigma$, is it always the case that $\I(\rho:\rho)\gea\I(\rho:\sigma)$? Is it true that every quantum state is typical of $\bmu$? For all density matrices $\rho$, is it the case that $\d(\rho|\bmu)=O(1)$? Since $\d(\rho|\bmu)^2\lea\I(\rho:\rho)$, we have that $\d(\rho|\bmu)\lea \sqrt{2n}$. Conservation of randomness and information has been proven to hold over unitary transforms and partial traces. Can conservation of randomness and information be proven with respect to quantum operations?

\bibliographystyle{plainnat}


\begin{thebibliography}{13}
\providecommand{\natexlab}[1]{#1}
\providecommand{\url}[1]{\texttt{#1}}
\expandafter\ifx\csname urlstyle\endcsname\relax
  \providecommand{\doi}[1]{doi: #1}\else
  \providecommand{\doi}{doi: \begingroup \urlstyle{rm}\Url}\fi

\bibitem[Benatti et~al.(2006)Benatti, Kr{\"u}ger, M{\"u}ller,
  Siegmund-Schultze, and Szko{\l}a]{Benatti06}
F.~Benatti, T.~Kr{\"u}ger, M.~M{\"u}ller, R.~Siegmund-Schultze, and
  A.~Szko{\l}a.
\newblock {Entropy and Quantum Kolmogorov Complexity: A Quantum Brudno's
  Theorem}.
\newblock \emph{Communications in Mathematical Physics}, 265\penalty0 (2),
  2006.

\bibitem[{Benatti} et~al.(2014){Benatti}, {Oskouei}, and {Deh Abad}]{Benatti14}
F.~{Benatti}, S.~K. {Oskouei}, and A.~S. {Deh Abad}.
\newblock {Gacs Quantum Algorithmic Entropy in Infinite Dimensional Hilbert
  Spaces}.
\newblock \emph{{Journal of Mathematical Physics}}, 55\penalty0 (8), 2014.

\bibitem[{Berthiaume} et~al.(2001){Berthiaume}, {van Dam}, and
  Laplante]{BerthiaumeVaLa01}
A.~{Berthiaume}, W.~{van Dam}, and S.~Laplante.
\newblock {Quantum Kolmogorov Complexity}.
\newblock \emph{Journal of Computer and System Sciences}, 63\penalty0 (2),
  2001.

\bibitem[Bu\ifmmode~\check{z}\else \v{z}\fi{}ek and Hillery(1996)]{BuzekHi96}
V.~Bu\ifmmode~\check{z}\else \v{z}\fi{}ek and M.~Hillery.
\newblock {Quantum Copying: Beyond the No-Cloning Theorem}.
\newblock \emph{Phys. Rev. A}, 54\penalty0 (3), 1996.

\bibitem[Chaitin(1975)]{Chaitin75}
G.~Chaitin.
\newblock {A Theory of Program Size Formally Identical to Information Theory}.
\newblock \emph{Journal of the ACM}, 22\penalty0 (3), 1975.

\bibitem[{G{\'a}cs}(2001)]{Gacs01}
P.~{G{\'a}cs}.
\newblock {Quantum Algorithmic Entropy}.
\newblock \emph{Journal of Physics A Mathematical General}, 34\penalty0 (35),
  2001.

\bibitem[Kolmogorov(1965)]{Kolmogorov65}
A.~Kolmogorov.
\newblock {Three Approaches to the Quantitative Definition of Information.}
\newblock \emph{{Problems in Information Transmission}}, 1\penalty0 (1), 1965.

\bibitem[Levin(1984)]{Levin84}
L.~Levin.
\newblock {Randomness Conservation Inequalities; Information and Independence
  in Mathematical Theories}.
\newblock \emph{{Information and Control}}, 61\penalty0 (1), 1984.

\bibitem[Li and Vit{\'a}nyi(2008)]{LiVi08}
M.~Li and P.~Vit{\'a}nyi.
\newblock \emph{{An Introduction to Kolmogorov Complexity and its
  Applications}}.
\newblock Springer Publishing Company, Incorporated, 3 edition, 2008.

\bibitem[{Nies} and {Scholz}(2018)]{NiesSc18}
A.~{Nies} and V.~{Scholz}.
\newblock {Quantum Martin-L{\"o}f randomness}.
\newblock \emph{ArXiv e-prints}, arXiv:quant-ph/1709.08422, 2018.

\bibitem[Solomonoff(1964)]{Solomonoff64}
R.~Solomonoff.
\newblock {A Formal Theory of Inductive Inference, Part l}.
\newblock \emph{Information and Control}, 7\penalty0 (1), 1964.

\bibitem[{Vitanyi}(2000)]{Vitanyi00}
P.~{Vitanyi}.
\newblock {Three Approaches to the Quantitative Definition of Information in an
  Individual Pure Quantum State}.
\newblock In \emph{Proceedings 15th Annual IEEE Conference on Computational
  Complexity}, 2000.

\bibitem[Zvonkin and Levin(1970)]{ZvonkinLe70}
A.~Zvonkin and L.~Levin.
\newblock {The Complexity of Finite Objects and the Development of the Concepts
  of Information and Randomness by Means of the Theory of Algorithms}.
\newblock \emph{Russian Math. Surveys}, 25\penalty0 (6), 1970.

\end{thebibliography}

\end{document}